\documentclass[10pt,conference]{IEEEtran}
\ifCLASSINFOpdf
\else
\fi
\usepackage{amsmath}%
\usepackage{amsfonts}%
\usepackage{amssymb}%
\usepackage{graphicx}%

\newtheorem{theorem}{Theorem}

\newtheorem{lemma}[theorem]{Lemma}

\newenvironment{proof}[1][Proof]{\textbf{#1.} }{\ \rule{0.5em}{0.5em}}

\newcommand{\Tmin}{T_{min}}
\newcommand{\Tmax}{T_{max}}

\begin{document}

\title{Proactive Resource Allocation: Turning Predictable Behavior into Spectral Gain}

\author{
\authorblockN{Hesham El Gamal}
\authorblockA{Department of Electrical\\
 and Computer Engineering\\
Ohio State University, Columbus, USA\\
helgamal@ece.osu.edu}
\and
\authorblockN{John Tadrous}
\authorblockA{Wireless Intelligent Networks\\
 Center (WINC)\\
Nile University, Cairo, Egypt\\
john.tadrous@nileu.edu.eg}
\and
\authorblockN{Atilla Eryilmaz}
\authorblockA{Department of Electrical\\
 and Computer Engineering\\
Ohio State University, Columbus, USA\\
eryilmaz@ece.osu.edu }
}

\maketitle

\begin{abstract}
This paper introduces the novel concept of proactive resource allocation in which the predictability of user behavior is exploited to balance the wireless traffic over time, and hence, significantly reduce the bandwidth required to achieve a given blocking/outage probability. We start with a simple model in which the smart wireless devices are assumed to predict the arrival of new requests and submit them to the network $T$ time slots in advance. Using tools from large deviation theory, we quantify the resulting {\bf prediction diversity gain} to establish that the decay rate of the outage event probabilities increases linearly with the prediction duration $T$. This model is then generalized to incorporate the effect of prediction errors and the randomness in the prediction lookahead time $T$. Remarkably, we also show that, in the cognitive networking scenario, the appropriate use of proactive resource allocation by the primary users results in more spectral opportunities for the secondary users at a marginal, or no, cost in the primary network outage. Finally, we conclude by a discussion of the new research questions posed under the umbrella of the proposed {\bf proactive (non-causal) wireless networking} framework.
\end{abstract}

\section{A New Paradigm for Resource Allocation}
\label{sec:Paradigm}
Ideally, wireless networks should be optimized to deliver the best Quality of Service (in terms of reliability, delay, and throughput) to the subscribers with the minimum expenditure in resources. Such resources include transmitted power, transmitter and receiver complexity, and allocated frequency spectrum. Over the last few years, we have experienced an ever increasing demand for wireless spectrum resulting from the adoption of {\em throughput hungry} applications in a variety of civilian, military, and scientific settings. Since the available spectrum is non renewable and limited, this demand motivates the need for efficient wireless networks that {\bf maximally utilize} the spectrum. In this work, we focus our attention on the resource allocation aspect of the problem and propose a new paradigm that offers remarkable spectral gains in a variety of relevant scenarios. More specifically, our proactive resource allocation framework exploits the predictability of our daily usage of wireless devices to smooth out the traffic demand in the network, and hence, reduce the required resources to achieve a certain point on the Quality of Service (QoS) curve. This new approach is motivated by the following observations.


\begin{itemize}
\item While we are experiencing a severe shortage in the spectrum, it is well-documented now that a significant fraction of the available spectrum
is under-utilized \cite{FCC2002}. This, in fact, is the main
motivation for the cognitive networking framework where secondary
users are allowed to use the spectrum in the off time, where the
primary users are idle, in an attempt to maximize the spectral
efficiency~\cite{Mitola}. Unfortunately, the cognitive radio
approach is still facing significant regulatory and technological
hurdles~\cite{Ian},~\cite{Gridlock} and, at best, will offer only a
partial solution to the problem. This limitation of the cognitive
radio approach is intimately tied to the main reason behind the
under-utilization of the spectrum; namely \emph{the large disparity
between the average and peak traffic demand in the network}. As an
example, if we take a typical cellular network, one can easily see
that the traffic demand in the peak hours is much higher than that
at night; which inspires the different rates offered by
cellular operators. Now, the cognitive radio approach assumes that
the secondary users will be able to utilize the spectrum in the off-peak times but, unfortunately, at those particular times one may
expect the secondary traffic characteristics to be similar to that
of the primary users (e.g., at night most of the primary and
secondary users are expected to be idle). As argued in the following, the overarching goal of the proactive resource allocation framework is to avoid this limitation, and hence, achieve a significant reduction in the peak to average demand ratio {\bf without relying on out of network users}.

\item In the traditional approach, wireless networks are
constructed assuming that the subscribers are equipped with {\em
dumb terminals} with very limited computational power. It is obvious
that the new generation of {\em smart devices} enjoy significantly
enhanced capabilities in terms of both {\bf processing power and
available memory}. Moreover, according to Moore's law predictions,
one should expect the computational and memory resources available
at the typical wireless device to increase at an exponential rate.
This observation should inspire a similar paradigm shift in the
design of wireless networks whereby the capabilities of the smart
wireless terminals are leveraged to maximize the utility of the
frequency spectrum, {\em a non-renewable resource that does not
scale according to Moore's law}. Our proactive resource allocation framework is a significant step in this direction.

\item The introduction of smart phones has resulted in a paradigm shift in the dominant traffic in mobile
cellular networks. While the primary traffic source in traditional
cellular networks was {\bf real time} voice communication, one can
argue that a significant fraction of the traffic generated by the
smart phones results from non-data-requests (e.g., file
downloads). As demonstrated in the following, this feature allows for more degrees of freedom in the design of the scheduling algorithm.

\item The final piece of our puzzle relates to the observation
that our usage of the wireless devices is {\bf highly predictable}.
This claim is supported by a growing body of evidence that range
from the recent launch of {\bf Google Instant} to the interesting
findings on our predictable mobility patterns~\cite{SQBB10}. In
our context, a relevant example would be the fact that our
preference for a particular news outlet is not expected to change
frequently. So, if the smart phone observes that the user is
downloading CNN, for example, in the morning for a sequence of days
in a row then it can {\bf safely anticipate} that the user will be
interested in the CNN again the following day. Coupled with the
fact that the most websites are refreshed at a relatively slow rate,
as compared with the dynamics of the underlying wireless network,
one can now see the potential for scheduling early downloads of the
predictable traffic to {\bf reduce the peak to average traffic
demand} by maximally exploiting the available spectrum in the network
idle time. 
\end{itemize}
It is important to observe here the {\bf temporal and spatial scales} at which this predictability phenomenon exhibits itself. First there is a growing body of evidence that our behavioral patterns can be accurately predicted at the {\bf single} user level. On the temporal scale, the requests are largely predictable at the scale of the application layer (e.g., minutes and hours) which is much slower than the dynamics of the physical, medium access, and network layers. This critical property is a key enabler for exploiting capacity enhancing techniques that introduce delays at the same time scale.

The objective of this paper is to highlight the potential improvement in the spectral efficiency of wireless networks through the judicious exploitation of the predictable behavior of wireless users. More specifically, in the current paradigm, traffic requests are considered urgent, at the time scale of the application layer, and hence, have to be served upon initiation by the network users in order to satisfy the required QoS metrics. However, if the wireless devices can {\bf predict} the requests to be generated by the corresponding users and submit them in advance, then the network will have the flexibility in scheduling these requests over an expanded time horizon as long as the imposed deadlines are not violated. When a {\bf predictive} network serves a request before its deadline, the corresponding data is stored in cache memory of the wireless device and, when the request is actually initiated, the application pulls the information directly from the memory instead of accessing the wireless network. It is worth noting that, not all applications, although predictable, can be served prior to their time of initiation. For example, some multimedia traffic maybe predictable, but, can only be served on a real time basis as they are based on live interactions between users. However, predicting these type of requests can still be considered as an advantage, as the network may schedule other non-real-time requests while taking into account the predicted real-time requests in a way that enhances the QoS of all applications.

The rest of this paper is mostly devoted to developing quantitative evidence that supports the previous qualitative discussion via analyzing certain asymptotic scenarios. More specifically, Section~\ref{sec:SysMod} describes a simplified system that will be the basis of our analytical results. The notion of {\bf prediction diversity} is introduced in Section~\ref{sec:Out_Anal} and quantified under different assumption on the performance of the prediction algorithm. Our analysis is extended to the scenario where users require different QoS guarantees, e.g., primary and secondary users in Section \ref{sec:Out_Sec}. Here, we demonstrate a remarkable phenomenon whereby prediction at one user, i.e., {\em good citizen}, is shown to improve the performance of the other without compromising its own. Throughout the paper, our theoretical claims are supported by numerical results that clearly illustrate the potentially remarkable gains in spectral efficiency that can be achieved by our proactive resource allocation approach. Finally, the paper is concluded in Section \ref{sec:Conc} with a discussion on the more general {\bf proactive wireless networking} paradigm and the research challenges associated with it.

\section{System Model}
\label{sec:SysMod}

Unless otherwise stated, we adopt a simplified model of a single cell slotted wireless network where the {\bf aggregate} requests are allowed to arrive only at the beginning of each slot. The number of arriving requests at time slot $n>0$ is denoted by $Q(n)$ which is assumed to be ergodic and to follow a Poisson distribution with rate $\lambda$. All requests are assumed to have the same amount of required resources which is taken to be unity. That is, each request has to be totally served in a single slot by consuming one unit of resource. Moreover, the wireless network has a {\bf fixed} capacity $C$ per slot. Furthermore, we assume that a predictive wireless network can anticipate the arrival of each request by an integer number of time slots in advance. That is, if $q(n)$, $1\leq q\leq Q(n)$, is the ID of a request predicted at the beginning of time slot $n$, the predictive network has the capability of serving this request no later than the next $T_{q(n)}$ slots. Hence, when a request $q(n)$ arrives at a predictive network, it has a deadline at time slot $D_{q(n)}=n+T_{q(n)}$ as shown in Fig. \ref{fig:Model_1}.
\begin{figure}
	\centering
		\includegraphics[width=0.40\textwidth]{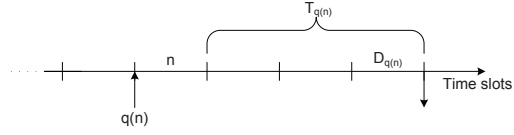}
	\caption{Prediction Model}
	\label{fig:Model_1}
\end{figure}
In the \emph{non-predictive} network, all arriving requests at the beginning of time slot $n$ have to be served in the same time slot $n$, i.e., if $q(n)$ is a non-predicted request, its deadline is $D_{q(n)}=n$ meaning that $T_{q(n)}=0$. We assume that an outage event occurs at a certain time slot if and only if at least one of the requests in the system expires in this slot. At this point, we wish to stress the fact our model operates as the time scale of the application layer at which 1) the current paradigm, i.e., non-predictive networking, treats all the requests as urgent, 2) each slot duration will be in the order of minutes and possibly hours, and 3) the system capacity is fixed since the channel fluctuation dynamic are averaged out at this time scale.

In this work, we study the probability of outage, $P(\text{outage})$, as the performance metric under a scaling regime whereby $\lambda$ and $C$ increase such that the ratio $\frac{\log(\lambda)}{\log C}$ is kept at a constant value $\gamma$, $0\leq\gamma \leq 1$. In other words, we scale $\lambda$ as $C^{\gamma}$ for each choice of $\gamma$. Under this assumption we characterize the \textbf{diversity gain} defined as 
\begin{equation}
d(\gamma)\triangleq \lim_{C\rightarrow\infty}\frac{-\log P(\text{outage})}{C \log C}
\end{equation}
for both the non-predictive and predictive networks.

\section{Prediction Diversity}
\label{sec:Out_Anal}
In this section we characterize the diversity gain for the two networks when both witness the same arrival process $Q(n)$, $n>0$ per slot. The difference only is in the deadlines of the arriving requests. The deadline for a request $q(n)$ is slot $n$ when the network is non-predictive, and is $n+T_{q(n)}$ when the network is predictive with $T_{q(n)}=1,2,\cdots$. In general, as the system capacity $C$ grows, the outage probability is expected to decrease. In our analysis we use tools of large deviation theory \cite{Gallager}, \cite{Glynn} to characterize $d(\gamma)$, which quantifies the achievable diversity-multiplexing tradeoff, in different scenarios. The following result determines the prediction diversity gain for the deterministic look-ahead time case, i.e., $T_{q(n)}=T$  $\forall q(n).$

\begin{theorem}
\label{th:1}
The diversity gain of proactive scheduling for the above model with
$T$-slot prediction equals
$$d_P(\gamma)=(1+T)(1-\gamma).$$
Noting that the diversity gain of the non-predictive scenario is
obtained as a special case by setting $T=0$, i.e.,
$d_N(\gamma)=(1-\gamma),$ this result reveals that proactive
scheduling offers a multiplicative gain of $(1+T)$ in the achievable
diversity advantage.
\end{theorem}
\begin{proof} (Sketch) We start with the non-predictive
benchmark corresponding to $T=0$. In this case, the outage
probability in any slot $n$ corresponds to the event $\{Q(n)>C\},$
which can be expressed as
\begin{equation}
\label{eq:P_N1}
P_N(\text{outage})=\sum_{k=C+1}^{\infty}{\frac{(C^{\gamma})^k}{k!}e^{-C^{\gamma}}}.
\end{equation}
For large values of $C$, tightest Chernoff bound \cite{Gallager} can be used to upper bound the outage probability as
\begin{equation}
\label{eq:C_N2}
P_N(\text{outage})\leq  e^{C-C^{\gamma}-(1-\gamma)C\log C}.
\end{equation}
Furthermore, from \eqref{eq:P_N1}, it is obvious that
\begin{equation}
\label{eq:L_N}
P_N(\text{outage})\geq \frac{C^{\gamma(C+1)}}{(C+1)!}e^{-C^{\gamma}},
\end{equation}
hence, by taking the $\log$ of the upper and lower bounds on $P_N(\text{outage})$ in \eqref{eq:C_N2}, \eqref{eq:L_N} and dividing by $-C\log C$ it follows directly that the diversity gain of the non-predictive network is equal to
\begin{equation}
\label{eq:d_N}
d_N(\gamma)=1-\gamma.
\end{equation}
For $T>0$, it is easy to see that the First-In-First-Out (FIFO), or
equivalently Earliest Deadline First (EDF), scheduling policy
minimizes the outage probability in this simple scenario. To
characterize the diversity gain, we first need to define the
following two events to upper and lower bound the outage event
$$\mathcal{U}_d(n)\triangleq\left\{\sum_{i=n-2T}^{n-T}{Q(i)}>C(T+1)\right\},$$
$$\mathcal{L}_d(n)\triangleq\left\{Q(n-T)>C(T+1) \right\}.$$ In
the steady state, i.e., when $n\rightarrow\infty$, we have shown
in~\cite{GamTadEry10} that $${\rm Pr}(\mathcal{L}_d(n))\leq
P_P(\text{outage})\leq {\rm Pr}(\mathcal{U}_d(n)).$$ We further showed
that $$\lim_{C\rightarrow\infty} -\frac{\log {\rm Pr}(\mathcal{L}_d)}{C
\log C}=\lim_{C\rightarrow\infty} -\frac{\log {\rm Pr}(\mathcal{U}_d)}{C
\log C}=(1+T)(1-\gamma).$$ Combining these two relationships results
in our claimed diversity gain expression:
$d_P(\gamma)=(1+T)(1-\gamma).$
\end{proof}

Now, we consider a more general case where $T_{q(n)}, 0\leq q\leq Q(n), n>0$ is a sequence of i.i.d. nonnegative integer-valued random variables defined over a finite support $\Tmin, \Tmin+1, \cdots, \Tmax$. First, we start with the scenario where probability mass function (PMF) of $T_{q(n)}$ does not scale with $C$ and establish the critical dependence of the achievable diversity gain on $\Tmin>0$.

\begin{lemma}
Let the PMF of $T_{q(n)}$ be given by
\begin{equation}
\label{eq:PMFofT}
Pr(T_{q(n)}=k)\triangleq\left\{
\begin{array}{l l}
p_k,\quad \quad \Tmin\leq k \leq\Tmax,\\
0,\quad \quad \text{otherwise},
\end{array}
\right.
\end{equation}
and the probabilities $p_k$'s are constants that do not depend on $C$. Then,
$$d_P(\gamma)=(1+\Tmin)(1-\gamma).$$
\end{lemma}
\begin{proof} (Sketch) A lower bound on the outage probability can be obtained by considering only the fraction of the requests corresponding to $\Tmin$ whereas an upper bound can be obtained by making $T_{q(n)}=\Tmin$ $\forall q(n)$. It is easy to see that both bounds have the same decay rate corresponding to the stated diversity advantage.
\end{proof}

It is clear that the diversity gain of random $T$ scenario is dominated by the requests with $T=\Tmin$, and hence, under the previous assumptions the system will not experience any prediction diversity gains when $\Tmin=0$. Two observations are now in order.

\begin{enumerate}

\item Despite the lack of gain in prediction diversity in this scenario, our numerical results, reported later, still demonstrate remarkable gains in {\bf the outage probability} for a wide range of system parameters.

\item When the fraction of requests corresponding to $\Tmin$ decays as $C$ grows, which is reasonable to expect in many emerging applications as most of the new demand corresponds to {\bf predictable and delay tolerant} data traffic, then the proactive resource allocation framework is able to harness improved prediction diversity gains. This can be viewed as follows. To illustrate the idea, let's assume that $\Tmin=0$ and $p_{\Tmin}=p_{0}=C^{-\alpha}$, $\alpha>0$, then it is easy to see that the diversity gain of the predictive network will be given by,
\begin{equation}
\label{eq:d_P_alpha1}
d_P(\gamma)=1+\alpha-\gamma
\end{equation}
as long as $1+\alpha-\gamma$ is smaller than $2(1-\gamma)$ or equivalently, $\alpha\leq 1-\gamma$. Otherwise, the diversity gain will be determined by the requests with $T=1$ and will be given by
\begin{equation}
\label{eq:d_P_alpha2}
d_P(\gamma)=2(1-\gamma).
\end{equation}
This argument is extended in~\cite{GamTadEry10} for more general distributions of the look-ahead time $T$.

\end{enumerate}

Thus far, we have shown that the proposed proactive resource allocation paradigm will significantly enhance the prediction diversity gain under the assumption of perfect, i.e., error free, prediction. Now, we investigate the effect of prediction error on the prediction diversity gain. In our analysis, we consider the deterministic $T$ scenario, and assume that the traffic of the non-predictive system is characterized by the process $Q(n), n>0$ which represents the number of arriving requests at the beginning of time slot $n$ with $T=0$. This process is Poisson with rate $C^{\gamma}$. Moreover, the system is operating according to the Shortest Deadline First scheduling policy. Our model differentiates between the following two prediction error events.
\begin{enumerate}
	\item The network mistakenly predicts a request and serves it resulting in an increase in the traffic load.
	\item The predictive network fails to predict a request and, as a consequence, it encounters an urgent arrival with $T_{q(n)}=0$.
\end{enumerate}
Therefore, the arriving requests $Q^E(n)$, $n>0$ can be regarded as the superposition of two arrival processes: 1) $Q'(n)$ corresponding to the the predicted request at the beginning of time slot $n$ with deadline $n+T$ and 2) $Q''(n)$ corresponding to the urgent requests arriving requests at the beginning of time slot $n$ and must be served instantaneously. {\bf The judicious design of the prediction algorithm should aim to strike the optimal balance between these two events}. This point is illustrated in the following special case: $Q'(n)$ is Poisson with rate $C^{\gamma'}$, where $\gamma'\in\Re$, and $Q''(n)$ is Poisson with rate $C^{\gamma''}$, $\gamma''\leq \gamma$ such that
\begin{equation}
\label{eq:Rate_constraint}
C^{\gamma'}+C^{\gamma''}\geq C^{\gamma}.
\end{equation}
The constraint $\gamma''\leq \gamma$ follows directly from the fact that the arrival rate of the urgent requests cannot exceed the arrival rate of requests in the error free scenario. On the other hand, the constraint \eqref{eq:Rate_constraint} reflects the fact prediction errors can only increase the arrival rate. In this model, a necessary and sufficient condition for perfect prediction is $\gamma'=\gamma$ and $\gamma''=-\infty$ resulting in $Q^E(n)=Q'(n)=Q(n+T)$. We also let the lookahead time $T$ to be a function of $(\gamma',\gamma'')$ reflecting the fact that more aggressive prediction algorithms will result in a larger $T$ at the expense of introducing larger prediction errors. Finally, we assume that, given $\gamma'$ and $\gamma''$, both processes $Q'(n)$ and $Q''(n)$ are independent.

By setting $\gamma'=\alpha' \gamma$ and $\gamma''=\alpha'' \gamma$, the diversity gain of the predictive network will be given by\footnote{Following similar analysis to that of Section \ref{sec:Out_Anal}.}
\begin{equation}
\label{eq:d_error0}
d_P(\gamma)=\min\{(1+T(\alpha',\alpha''))(1-\max\{\alpha',\alpha''\} \gamma),1-\alpha'' \gamma\}.
\end{equation}
If $\max\{\alpha',\alpha''\}=\alpha''$ the diversity of the predictive network becomes $d_P(\gamma)=1-\alpha''\gamma$. However, since $\alpha''\leq 1$ and from \eqref{eq:Rate_constraint}, it is straightforward to see that $\max\{\alpha',\alpha''\}=\alpha''$ if and only if $\alpha'=\alpha''=1$ corresponding to the scenario where the predictive mechanism is useless. Therefore, in the following we focus on the case where $\alpha'\geq \alpha''$ in which case the prediction diversity gain is given by
\begin{equation}
\label{eq:d_error}
d_P(\gamma)=\min\{(1+T(\alpha',\alpha''))(1-\alpha'\gamma),1-\alpha''\gamma\},
\end{equation}

implying that the predictive system achieves a \emph{strictly} improved diversity gain over the non-predictive system if and only if,
\begin{equation}
\label{eq:Improved_Diver}
\min\{(1+T(\alpha',\alpha''))(1-\alpha' \gamma),1-\alpha'' \gamma\}>1-\gamma.
\end{equation}

We further note that an upper bound on the prediction diversity, for a given $(\alpha', \alpha'')$, corresponds to case where the optimum operating point for the the two quantities inside the $\min\{.\}$ are equal, i.e.,
\begin{equation}
\label{eq:Op_Point}
(1+T(\alpha',\alpha''))(1-\alpha' \gamma)=(1-\alpha'' \gamma)
\end{equation}
or
\begin{equation}
\label{eq:T^*}
T(\alpha',\alpha'')=\frac{(\alpha'-\alpha'')\gamma}{1-\alpha' \gamma}.
\end{equation}

Hence, for a given $(\alpha', \alpha'')$, a prediction algorithm that achieves \eqref{eq:T^*} is optimal in terms of the achievable prediction diversity and there will be no benefit in increasing $T$ further. Based on that, we can see that the achievability of prediction diversity gains hinges on the existence of prediction algorithms that satisfy the following necessary conditions

\begin{align}
&1\leq \alpha'\leq \frac{1}{\gamma},\\
&\alpha''<1
\end{align}

At this point, we wish to stress the fact that the previous model for prediction errors was intended only to illustrate the tradeoff between the two types of error events identified earlier. Our current investigations aim at developing more accurate models that reflect the nature of the traffic requests and the dynamics of the employed prediction algorithms. We conclude this section with numerical results that illustrate the performance gain offered by the proposed proactive resource allocation framework. In Fig. \ref{fig:Paper_sim_FR_T} we plot the outage probability of predictive and non-predictive networks versus $C \log C$. The simulation is based on the EDF policy with $\gamma=0.8$. At each value of $C$, the system is simulated for $10^3$ time slots and the performance is averaged over $10^2$ simulation runs. It is clear, from the results, that there is a remarkable reduction in the resources required to attain a certain level of outage probability when the network employs the predictive resource allocation mechanism.
\begin{figure}
	\centering
		\includegraphics[width=0.45\textwidth]{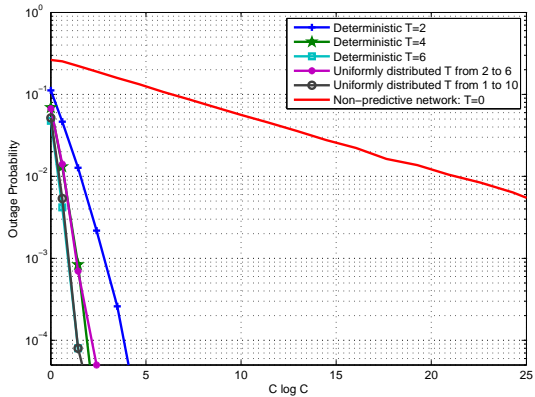}
	\caption{Outage probability vs. $C \log C$ with $\gamma=0.8$.}
	\label{fig:Paper_sim_FR_T}
\end{figure}
Moreover, for the two simulated random $T$ scenarios, although $\Tmin$ is chosen to be $2$ and $1$, the corresponding outage probability curves are upper bounded by the outage probability of the predictive case with deterministic $T=2$. This actually may be a consequence of the small values of $C$ in this figure. Here, the averaging effect over the range between $\Tmin$ and $\Tmax$ appears to have a more favorable impact on the performance than increasing $\Tmin$.

\begin{figure}
	\centering
		\includegraphics[width=0.45\textwidth]{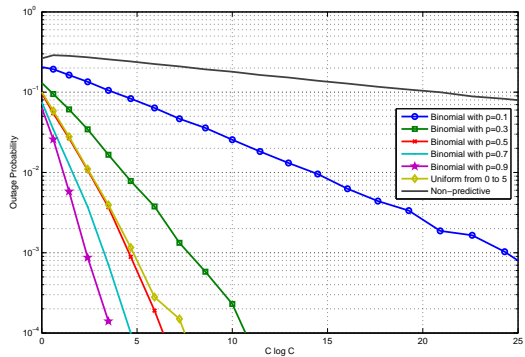}
	\caption{Effect of different distributions $T$ on the outage performance ($\gamma=0.9$).}
	\label{fig:Paper_sim_Dist}
\end{figure}
Fig.~\ref{fig:Paper_sim_Dist} investigates the effect of the distribution of $T$ on the outage (blocking) probability. Here, we consider a class of binomial distributions with finite support from $\Tmin=0$ to $\Tmax=5$ and parameter $p$. That is, $$Pr(T=t)=p_t=\binom{\Tmax}{t}p^t(1-p)^{\Tmax-t}$$ where $\Tmin\leq t\leq \Tmax$. The predictive system is then simulated for different values of $p$ and the outage probability results are depicted. Moreover, the uniform distribution of $T$ over the interval $\Tmin=0$ to $\Tmax=5$ is plotted on the same figure. From the results, one can argue that the outage performance is sensitive to the value of $p_{\Tmin}$ over the simulated range of $C$. Since the binomial distributions of $T$, $p_{\Tmin}$ is monotonically decreasing with $p$ and thus, as the weight of the arrivals with $T=0$ increases the outage behavior becomes worse although all of the outage curves have the same diversity gain in infinite $C$ asymptotic. Also, in case of a uniform distribution, the outage probability curve is quite close to that of the binomial distribution with $p=0.5$ although $p_{\Tmin}$ of the uniform is larger that its peer of the binomial with $p=0.5$. The reason behind this behavior is that, the weights of the higher values of $T$ in case of the uniform distribution are larger than their peers in case of the binomial distribution with $p=0.5$. This advantage enables the scheduler to efficiently reduce the outage probability despite the relatively large probability corresponding to $T=0$ in the uniformly distributed $T$.

\section{Different QoS Users: The Good Citizen Phenomenon}
\label{sec:Out_Sec}
The previous section demonstrates the potential gains that can be leveraged from the proactive resource allocation framework when all the requests belong to the same class of QoS. In this section we consider a network with two QoS classes that can be considered as primary and secondary users sharing the same resources. We investigate the effect of prediction by {\bf the primary user only} on the prediction diversity gain of the secondary network. Clearly, our analysis can be extended to allow for prediction by the secondary user as well; but we choose to limit ourselves to this special case for simplicity.
We assume that the number of secondary arrivals at the beginning of time slot $n$ is $Q^s(n)$, where $Q^s(n)$ follows a Poisson distribution with rate $\lambda^s=C^{\gamma^s}$, $0\leq\gamma^s\leq 1$. The number of primary requests arriving at the beginning of time slot $n$ is $Q^p(n)$ that follows a Poisson distribution with rate $\lambda^p=C^{\gamma^p}$, where $0\leq\gamma^p\leq 1$. We assume that the system is dominated by primary arrivals, that is, $\lambda^p>\lambda^s$ or, equivalently, $\gamma^p>\gamma^s$. The secondary and primary arrival processes are ergodic and independent.
\subsection{Non-Predictive Primary User}
We analyze the outage probability of the secondary user and its diversity gain when the primary user is non-predictive. At the beginning of time slot $n$, the system is supposed to witness $Q^p(n)+Q^s(n)$ arriving requests with deadline is slot $n$, i.e., must be served in the same slot of arrival.
The primary system has a fixed capacity $C$ per slot. In order to enhance the utilization of its resources, the primary user allows secondary requests to be served  by the remaining resources from serving the primary requests. Thus, at slot $n$, the remainder of $C-Q^p(n)$ is assigned to serve the secondary requests. The following result characterizes the achievable diversity gain in this scenario
\begin{theorem}
In the non-predictive scenario, the primary and secondary diversity are equal and given by
\begin{equation}
d_N^s(\gamma^p,\gamma^s)=d_N^p(\gamma^p,\gamma^s)=1-\gamma^p.
\end{equation}

\end{theorem}
\begin{proof} (Sketch)
The outage probability of the primary system $P_N^p(\text{outage})$ is identical to the one analyzed in the previous section. As a result, the primary diversity gain is given by
\begin{equation}
\label{eq:d_N^p}
d_N^p(\gamma^p,\gamma^s)=1-\gamma^p.
\end{equation}

The secondary system encounters an outage at a given slot when the remaining resources from serving the primary requests at this slot are less than the number of arriving secondary requests at the beginning of the same slot. Thus, if the primary network suffers an outage in a certain slot with at least one arriving secondary request, the secondary system goes in outage as well. The secondary system, consequently, encounters an outage at slot $n$ if and only if $$Q^p(n)+Q^s(n)>C \quad \text{and} \quad Q^s(n)>0.$$
Let the outage probability of the secondary network when the primary network is non-predictive be denoted by $P_N^s(\text{outage})$, hence
\begin{equation}
\label{eq:P_N^s1}
P_N^s(\text{outage})=Pr\left(Q^p(n)+Q^s(n)>C, Q^s(n)>0\right).
\end{equation}
The two random variables $Q^p(n)+Q^s(n)$ and $Q^s(n)$ are dependent but their joint distribution can simply be obtained by transformation of variables. By setting $Y=Q^p(n)+Q^s(n)$ and $U=Q^s(n)$, the exact expression of $P_N^s(\text{outage})$ will be given by
\begin{equation}
\label{eq:P_N^sExact}
\begin{split}
P_N^s(\text{outage})&=Pr(Y>C,U>0)\\
                    &=\sum_{y=C+1}^{\infty}\sum_{u=1}^{y}\frac{C^{{\gamma^p} (y-u)+{\gamma^s} u}}{(y-u)!u!}e^{-(C^{\gamma^p}+C^{\gamma^s})}.
\end{split}
\end{equation}
The diversity gain of the secondary system coexisting with a non-predictive primary network is defined by $$d_N^s(\gamma^p,\gamma^s)\triangleq \lim_{C\rightarrow\infty}\frac{-\log P_N^s(\text{outage})}{C\log{C}}.$$ For large values of $C$, the outer sum of the right hand side of \eqref{eq:P_N^sExact} is dominated by $y=C+1$. However, the inner sum is not dominated by a single value of $u$ because of $(y-u)!u!$ in the denominator. Consequently, as $C\rightarrow\infty$, $P_N^s(\text{outage})$ can be written as
\begin{equation}
\label{eq:P_N^s2}
P_N^s(\text{outage})\doteq\sum_{u=1}^{C+1}\frac{C^{{\gamma^p} (C+1-u)+{\gamma^s} u}}{(C+1-u)!u!}e^{-(C^{\gamma^p}+C^{\gamma^s})}.
\end{equation}
Characterizing $d_N^s(\gamma^p,\gamma^s)$ from \eqref{eq:P_N^s2} is, however, difficult, so we consider another approach based on the asymptotic behavior of upper and lower bounds on $P_N^s(\text{outage})$.

\subsubsection{Upper Bound on $P_N^s(\text{outage})$}
Since $Pr(\mathcal{A},\mathcal{B})\leq Pr(\mathcal{A})$ with equality if and only if $\mathcal{A}\subseteq{\mathcal{B}}$, then
\begin{align}
P_N^s(\text{outage})& \leq Pr(Q^p(n)+Q^s(n)>C).
\end{align}
The random variable $Q^p(n)+Q^s(n)$ has a Poisson distribution with mean $C^{\gamma^p}+C^{\gamma^s}$, then applying upper and lower bounds on $Pr(Q^p(n)+Q^s(n)>C)$ similar to that conducted with $Pr(Q(n)>C)$ in the proof of Thoerem \ref{th:1}, the diversity gain of the secondary network when the primary network is non-predictive is lower bounded by
\begin{align}
d_N^s(\gamma^p,\gamma^s) &=1-\max\{\gamma^p,\gamma^s\}\\
\label{eq:d_N^sL}
                         &=1-\gamma^p.
\end{align}

We consider the event that there is at least one secondary arrival with a primary outage at slot $n$ as a sufficient but not necessary condition on a secondary outage at slot $n$. That is, $$\mathcal{L}_N^s(n)\triangleq\left\{ Q^p(n)>C, Q^s(n)>0\right\}, \quad n\rightarrow\infty.$$
Note that, the event $\mathcal{L}_N^s(n)$ is not necessary for a secondary outage at slot $n$ as there may be $Q^p(n)<C$ but $Q^s(n)>C-Q^p(n)$ which results in a secondary outage at slot $n$ too. Furthermore, at steady state, $Pr(\mathcal{L}_N^s(n))$ becomes independent of $n$ as both arrival processes, $Q^p(n)$ and $Q^s(n)$, are stationary, hence we use $Pr(\mathcal{L}_N^s)$ instead.
Since $\mathcal{L}_N^s(n)$ is a sufficient condition for a secondary outage, then $P_N^s(\text{outage})\geq Pr(\mathcal{L}_N^s)$. Hence,
\begin{align}
P_N^s(\text{outage}) &\geq Pr(Q^p(n)>C,Q^s(n)>0)\\
                    &=Pr(Q^p(n)>C).Pr(Q^s(n)>0)\\
                    &=Pr(Q^p(n)>C)(1-C^{-\gamma^s})
\end{align}
Therefore
\begin{multline*}
\lim_{C\rightarrow\infty}\frac{-\log P_N^s(\text{outage})}{C\log C}\leq \\ \lim_{C\rightarrow\infty} \frac{-\log Pr(Q^p(n)>C)}{C\log C}-\frac{\log(1-C^{-\gamma^s})}{C\log C},
\end{multline*}
yielding
\begin{equation}
\label{eq:d_N^sU}
d_N^s(\gamma^p,\gamma^s)\leq 1-\gamma^p.
\end{equation}

From \eqref{eq:d_N^sL}, \eqref{eq:d_N^sU}, it follows that
\begin{equation}
d_N^s(\gamma^p,\gamma^s)=1-\gamma^p.
\end{equation}
Hence, it is obvious that the diversity gain of the secondary network in a primary non-predictive mode is the same as the diversity gain of the primary network although the arrival rate of secondary requests is strictly smaller than the primary arrival rate.
\end{proof}

\subsection{Predictive Primary User}
In this case, the system can {\bf only} predict the primary arrivals by $T$ time slots in advance. We assume that $T$ is deterministic and fixed for all primary requests, i.e., the deadline for the primary requests $Q^p(n)$ is $n+T$. The system, however, is assumed to be non-predictive for the secondary requests, i.e., the deadline for the secondary requests $Q^s(n)$ is $n$. When this system dedicates {\bf all} the per-slot capacity $C$ to serve the primary requests, according to the EDF policy, secondary requests arriving at the beginning of time slot $n$ will be served if and only if $C$ is strictly larger than the number of primary requests \emph{existing} in the system at the beginning of this slot. Unfortunately, this service policy does not enhance the outage performance of the secondary system although it minimizes the outage probability of the primary. The main reason is the large variations in the number of served primary requests per slot that takes on values from $0$ to $C$. These variations are quite close to the variations in the number of served primary requests per slot in case of non-predictive primary network.
\begin{figure}
	\centering
		\includegraphics[width=0.45\textwidth]{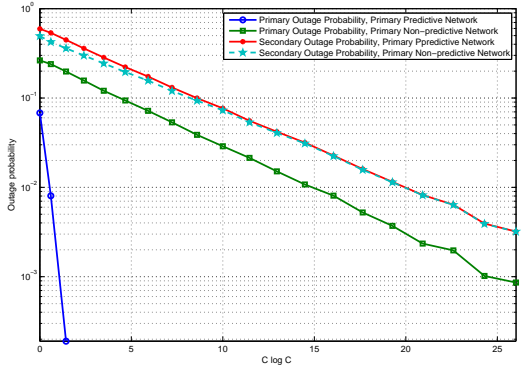}
	\caption{Outage probability vs. $C\log C$ for primary and secondary networks under the two types of primary network: predictive and non-predictive. All are calculated assuming SP1 ($\gamma^p=0.75$, $\gamma^s=0.05$ and $T=4$).}
	\label{fig:Notes3SP1Effect}
\end{figure}
Fig. \ref{fig:Notes3SP1Effect} plots the outage probability of the primary and secondary networks versus $C\log C$ under the two types of primary network, predictive and non-predictive. The results are based on simulations over $M=10^3$ slots and averaging over $100$ simulation runs. It is clear that the outage probability of the primary system when the primary network is predictive is significantly improved over its peer when the primary network is non-predictive. However, it can be noted that by {\bf selfishly} minimizing the outage probability of the primary network, one does not leave room for enhancing the outage probability of the secondary network. In the following, we describe two representative {\bf good citizen} primary policies that result in significant gains in the secondary outage probability at a very marginal cost in terms of the primary outage.

The main idea motivating the first service policy is to minimize the probability of the \emph{dominant} outage event instead of minimizing the overall outage probability. Thus, the diversity gain of the primary network will not be affected while creating more opportunities for secondary requests. Consequently, the outage probability of the secondary network will be enhanced at the same diversity gain of the primary network.

\noindent {\bf Service Policy 2 (SP2):} The primary network is assigned a fixed capacity per slot of $C-\left\lfloor C^{\beta}\right\rfloor$ where $\beta<1$. It uses this fixed capacity to serve as much as possible of primary requests in the system according to the shortest deadline request policy.\\

Clearly SP2 achieves the optimal primary diversity advantage, i.e., $d_P^p(\gamma^p)=(1+T)(1-\gamma^p)$. Moreover, it is shown, numerically, in the following that the outage probability of the secondary network is improved because of the dedicated capacity of $\left\lfloor C^{\beta}\right\rfloor$. At this point, we observe that SP2 allocates a fixed capacity per slot to the primary network. However, due to the variability of the arrival process, one may expect some performance gains if the service policy adaptively decides on the allocated capacity for the primary network based on the number of requests in the system at each slot and their deadlines. This intuition motivates the following policy

\noindent {\bf Service Policy 3 (SP3):} Let $N^p(n)$ be the number of the primary requests in the system at the beginning of time slot $n$, and $N_d^p(n)$ be the number of these requests whose deadline is slot $n$. Then, the capacity of the primary network at slot $n$ is calculated as $$\min\left\{C, N_d^p(n)+f\times(N^p(n)-N_d^p(n))\right\}$$ where $0\leq f\leq1$. After that, the network serves the primary requests according to the EDF policy.\\

It is obvious that the performance of SP3 is highly dependent on the design parameter $f$. At $f=0$, the system, at steady state, is serving only the requests whose deadline is the current slot. In this case the system will be similar to the non-predictive network in terms of primary and secondary outage probabilities. At $f=1$, the system is very selfish, and hence, achieving the optimal primary outage probability. The following numerical results, however, show that intermediate values for $f$ result in significant improvement in the secondary outage while keeping the primary outage probability almost indistinguishable from the optimal one.

The performance of a network with primary and secondary users has been evaluated numerically with the same parameters of Fig. \ref{fig:Notes3SP1Effect} and the results are reported in Figs. \ref{fig:Paper_SP2}, \ref{fig:Paper_SP3}. In Fig. \ref{fig:Paper_SP2}, the outage probability of a primary network following SP2 with $\beta=0.3$ is shown. It is clear from the figure that the outage probability of the secondary network is enhanced over the non-predictive case. However, this improvement comes at the expense of shifting the outage probability curve of the primary network to the right while preserving the optimal diversity advantage. Moreover, although improved, the outage probability of the secondary network appears to have no gain in the decay rate (i.e., diversity).

\begin{figure}[t]
	\centering
		\includegraphics[width=0.45\textwidth]{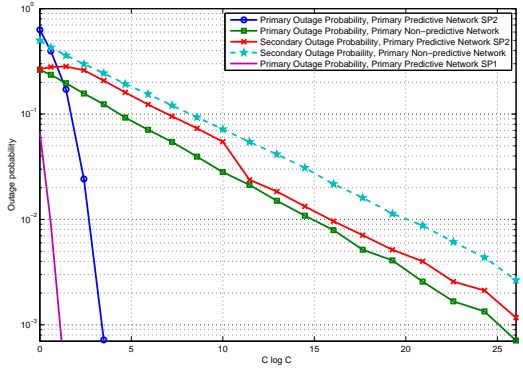}
	\caption{Outage probability vs. $C \log C$ of the primary and secondary users with $\gamma^p=0.75$, $\gamma^s=0.05$, $T=4$ and $\beta=0.3$.}
	\label{fig:Paper_SP2}
\end{figure}

\begin{figure}[t]
	\centering
		\includegraphics[width=0.45\textwidth]{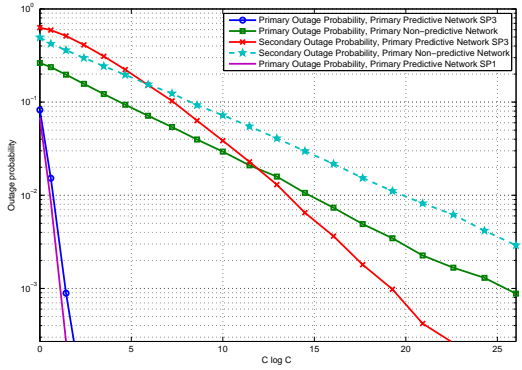}
	\caption{Outage probability vs. $C \log C$. of the primary and secondary users with $\gamma^p=0.75$, $\gamma^s=0.05$, $T=4$ and $f=0.5$.}
	\label{fig:Paper_SP3}
\end{figure}

In Fig. \ref{fig:Paper_SP3}, SP3 is evaluated for $f=0.5$. Compared with SP2, the behavior of SP3 is shown to remarkably enhance the outage probability of the secondary user at an almost negligible loss in the primary outage performance. The analysis of the diversity gain of the primary and secondary users operating according to SP3, however, are still under investigation. Overall, it can be concluded that, prediction at the primary side only does not {\bf only} enhance the primary spectral efficiency, but it can be efficiently exploited to significantly improve the spectral efficiency of the coexisting non-predictive secondary users (networks) as well.

\section{Conclusions}
\label{sec:Conc}
We have proposed a novel paradigm for wireless resource allocation which exploits the predictability of user behavior to minimize the spectral resources (e.g., bandwidth) needed to achieve certain QoS metrics. Unlike the tradition reactive resource allocation approach, in which the network can only start serving a particular user request upon its initiation, our proposed resource allocation approach anticipates future requests which allows the network more flexibility in scheduling those potential requests over an extended period of time. By adopting the outage (blocking) probability as our QoS metric, we have established the potential of our proactive resource allocation framework to achieve significant spectral efficiency gains in several interesting scenarios. More specifically, we introduced the notion of prediction diversity gain and used it to quantify the gain offered by the proposed resource allocation algorithm under different assumption on the performance of the traffic prediction technique. Moreover, we have shown that, in a network with two QoS classes, prediction at one side only does not only enhance its diversity gain, but it also improves the outage probability performance of the other user. Throughout the paper, our theoretical claims were supported by numerical results that illustrate the remarkable gains that can be leveraged from the proposed techniques.

We believe that this work has only scratched the surface of a very interesting research area which spans several disciplines and could potentially have a significant impact on the design of future wireless networks. In fact, one can immediately identify a multitude of interesting research problems at the intersection of information theory, machine learning, behavioral science, and networking. For example, our analysis have focused on the case of {\bf fixed supply and variable demand}. Clearly, the same approach can be used to {\bf match} demand with supply under more general assumptions on the two processes. In a different direction, our results should motivate further investigations on the design of efficient prediction algorithms; which will possibly require advanced tools from machine learning in addition to accurate models for user behavior that captures the predictability of traffic requests. Another avenue for future work is the cross layer optimization of content delivery over wireless networks under the proactive resource allocation models (i.e., the potential for multicast, peer-to-peer, and coupling between the time scales of different layers). Overall, as the need for wireless content delivery grows, we believe that the {\bf predictability} of the traffic pattern will be a {\bf key enabler} for exploiting a number of important factors to enhance the capacity of wireless networks. The basic idea is that, via the judicious use of predictability at {\bf the application layer time scale}, we will be able to design {\bf non-causal wireless networks}, from the end user perspective, that offer remarkable gains {\bf in capacity and delay}. Within this paradigm, our work can be viewed as the first step in laying the foundation for a systematic framework for the design and analysis of future proactive wireless networks.

\end{document}